\documentclass[journal]{IEEEtran}
\usepackage{fancyhdr,graphicx,amsmath,amssymb, amsfonts}
\usepackage{dsfont}
\usepackage{array}
\usepackage[caption=false,font=normalsize,labelfont=rm,textfont=sf]{subfig}
\usepackage{textcomp}
\usepackage{stfloats}
\usepackage{url}
\usepackage{verbatim}
\usepackage{graphicx}
\usepackage[nocompress]{cite}
\usepackage{xcolor}
\usepackage{amsthm}
\usepackage[ruled,vlined,linesnumbered]{algorithm2e}
\usepackage{nomencl}
\usepackage{algpseudocode}
\DeclareMathOperator*{\argmin}{arg\,min}
\usepackage{makecell}

\newtheorem{theorem}{Theorem}

\newtheorem{corollary}{Corollary}

\newtheorem{remark}{Remark}
\usepackage[colorlinks=true, linkcolor=blue, citecolor=blue, urlcolor=blue]{hyperref}
\begin{document}

\title{Binary Hypothesis Testing-Based Low-Complexity Beamspace Channel Estimation for mmWave Massive MIMO Systems}

\author{Hanyoung~Park,~\IEEEmembership{Graduate~Student~Member,~IEEE,}
        and~Ji-Woong~Choi,~\IEEEmembership{Senior~Member,~IEEE}
\thanks{This work was supported by Institute of Information \& communications Technology Planning \& Evaluation (IITP) grant funded by the Korea government (MSIT) (No. RS-2024-00442085, No. RS-2024-00398157). \textit{(Corresponding author: Ji-Woong Choi.)}}
\thanks{Hanyoung Park was with the Department of Electrical Engineering and Computer Science, Daegu Gyeongbuk Institute of Science and Technology, Daegu 42988, South Korea. He is now with Institute of Innovation for Future Army, Republic of Korea Army, Daejeon 34051, South Korea (e-mail: prkhnyng@dgist.ac.kr).}
\thanks{Ji-Woong Choi is with the Department of Electrical Engineering and Computer Science, Daegu Gyeongbuk Institute of Science and Technology, Daegu 42988, South Korea (e-mail: jwchoi@dgist.ac.kr).}
}

\markboth{To appear in IEEE Internet of Things Journal}%
{Shell \MakeLowercase{\textit{et al.}}: A Sample Article Using IEEEtran.cls for IEEE Journals}

\maketitle

\begin{abstract}
Millimeter-wave (mmWave) communications have gained attention as a key technology for high-capacity wireless systems, owing to the wide available bandwidth. 
However, mmWave signals suffer from their inherent characteristics such as severe path loss, poor scattering, and limited diffraction, which necessitate the use of large antenna arrays and directional beamforming, typically implemented through massive MIMO architectures. 
Accurate channel estimation is critical in such systems, but its computational complexity increases proportionally with the number of antennas. 
This may become a significant burden in mmWave systems where channels exhibit rapid fluctuations and require frequent updates. 
In this paper, we propose a low-complexity beamspace channel denoiser that exploits beamspace sparsity through Bayesian binary hypothesis testing, enabling direct component-wise signal detection without iterative optimization or learning-based training.
By modeling each sparse beamspace component as a mixture of signal and noise under a Bernoulli-complex Gaussian prior, we formulate a likelihood ratio test to detect signal-relevant elements.
Then, a hard-thresholding rule is applied to suppress noise-dominant components in the noisy channel vector.
Despite its extremely low computational complexity, the proposed method achieves channel estimation accuracy that is comparable to that of complex iterative or learning-based approaches.
This effectiveness is supported by both theoretical analysis and numerical evaluation, suggesting that the method can be a viable option for mmWave systems with strict resource constraints.
\end{abstract}

\section{Introduction} 
Millimeter-wave (mmWave) communication is widely regarded as a promising technology for future wireless networks, owing to its potential to support the growing demand for strict quality-of-service (QoS) requirements such as high data rates and dense user connectivity\cite{ref:mmwavepromising}.
This capability makes it well-suited for next-generation wireless applications that require both high capacity and low latency, such as virtual/augmented reality (VR/AR), smartglasses, mobile edge computing, autonomous vehicles, real-time industrial automation, and various Internet of Things (IoT) devices\cite{ref:ericsson, ref:mec, ref:v2x}.
Despite these benefits, mmWave communication faces several notable challenges. 
Due to their high frequency, mmWave signals suffer from severe pathloss, are highly susceptible to blockage, and exhibit limited diffraction\cite{ref:pathloss}.
Moreover, they tend to exhibit only a few dominant multipath components, resulting in sparse and highly directional channel responses in the spatial domain\cite{ref:mmwavesparse}. 
To overcome these propagation challenges and ensure reliable connectivity, mmWave systems typically adopt massive multiple-input multiple-output (MIMO) architectures to perform directional beamforming along dominant paths\cite{ref:mmwavetextbook}.

The effectiveness of such beamforming-based systems relies heavily on the availability of accurate channel state information (CSI), particularly in dynamic wireless environments characterized by user mobility, frequent blockages, and rapid channel variations\cite{ref:beamformingcsi}.
These dynamics are further exacerbated in mmWave systems due to their inherently short channel coherence time, which significantly limits the duration over which the channel remains stable\cite{ref:mmwavecoherence}. 
Consequently, channel estimation must be performed frequently and with minimal latency to remain effective.
Reliable channel estimation under these conditions is critical for a wide range of core functionalities, including user scheduling\cite{ref:app_userscheduling}, beam selection\cite{ref:app_beamselection}, precoding\cite{ref:app_precoding}, and equalization\cite{ref:app_equalization}.
As a result, the development of practical and accurate channel estimation techniques that operate efficiently under tight latency and overhead constraints remains a critical challenge in mmWave MIMO system design.

\subsection{Prior Works}\label{subsec:priorwork}
Previous works have proposed various channel estimation methods for mmWave MIMO systems.
These works can be broadly grouped into three directions.
The first line of research exploits beamspace sparsity of mmWave channels through compressive sensing (CS) frameworks \cite{ref:rwcs1, ref:rwcs2, ref:rwcs3, ref:rwcs4, ref:rwcs5, ref:rwcs6, ref:rwcs7, ref:rwcs8, ref:rwcs9, ref:rwcs10}.
The works in \cite{ref:rwcs1, ref:rwcs2, ref:rwcs3} proposed channel estimation frameworks using approximate message passing (AMP), which is a representative CS algorithm.
AMP-based channel estimation algorithms have also been developed for low-resolution receivers\cite{ref:rwcs4, ref:rwcs5}.
Meanwhile, orthogonal matching pursuit (OMP), another representative CS algorithm, has also been adopted for channel estimation\cite{ref:rwcs6, ref:rwcs7, ref:rwcs8, ref:rwcs9, ref:rwcs10}.
However, these algorithms require iterative optimization or are computationally intensive, which makes their adoption difficult in dynamic propagation environments that change rapidly.

The second research trend adopts deep learning for channel estimation\cite{ref:rwsl1,ref:rwsl2,ref:rwsl3,ref:rwsl4,ref:rwsl5,ref:rwsl6}.
The method in \cite{ref:rwsl1} proposed a supervised denoising model based on AMP and a denoising convolutional neural network (DnCNN) for beamspace channel estimation.
Ma \textit{et al.} devised a learned AMP-based sparse channel estimation framework, which reduces pilot overhead by jointly optimizing the pilot sequence and the channel estimator\cite{ref:rwsl2}.
Abdallah \textit{et al.} proposed a DnCNN-based channel estimator that learns channel support and amplitudes in the frequency domain\cite{ref:rwsl3}.
Wu \textit{et al.} applied deep unfolding for sparse beamspace channel estimation leveraging trimmed-ridge regression\cite{ref:rwsl4}.
Jang and Lee proposed the joint optimization of environment-specific beamformers and geometric parameters based on a neural network for wideband mmWave channel estimation\cite{ref:rwsl5}.
Zhang \textit{et al.} proposed a denoising approach using a generative adversarial network (GAN), which refines noisy channel estimates by leveraging spatial and temporal correlations in time-varying mmWave channels\cite{ref:rwsl6}.
However, aforementioned frameworks require paired ground-truth data, implying a strong dependence on potentially expensive datasets containing a large number of noiseless channel realizations for training.
To overcome this limitation, unsupervised learning-based channel estimation methods have been devised\cite{ref:rwul1,ref:rwul2,ref:rwul4,ref:rwul5,ref:rwul6,ref:deq,ref:ldgec}.
Tong \textit{et al}. proposed an unsupervised learning-based two-step OMP algorithm for channel estimation\cite{ref:rwul1}. 
The work in \cite{ref:rwul2} adopted subcarrier-wise estimation of signal parameters via rotational invariant techniques (ESPRIT) and $k$-means clustering to estimate the wideband mmWave channel.
Gahlot \textit{et al.} devised a training-free channel estimation that performs denoising using anisotropic diffusion model\cite{ref:rwul4}.
Balevi \textit{et al.} proposed a CS-based channel estimation with deep generative network\cite{ref:rwul5}.
Zhou \textit{et al.} utilized a diffusion model in MIMO channel estimation problem with Stein's unbiased risk estimate (SURE)\cite{ref:rwul6}.
He \textit{et al.} formulated wideband beamspace channel estimation as a CS problem on the 2-dimensional (2D) image using deep unfolding, and applied SURE to enable unsupervised learning\cite{ref:ldgec}.
In \cite{ref:deq}, the authors proposed deep equilibrium model (DEQ) for beamspace channel estimation, which eliminates the dependency on ground-truth data and the white noise assumption inherent in SURE by using generalized SURE (GSURE).
Nevertheless, these methods require training phases and have weak explainability.
Moreover, the substantial computational demands of deep learning-based methods may hinder their practicality in realistic wireless environments and hardware-constrained deployments.

Another line of research focuses on reducing computational complexity through lightweight denoisers and simplified frameworks, as the aforementioned methods can become computationally demanding when the number of antennas increases\cite{ref:rwlc1,ref:rwlc2,ref:rwlc3,ref:rwlc4,ref:LAMP-GM,ref:rwlc5,ref:rwlc6, ref:rwlc7, ref:beaches}.
The works in \cite{ref:rwlc1} and \cite{ref:rwlc2} proposed frameworks that reduce the complexity of sparse Bayesian learning-based channel estimation.
Zhou \textit{et al.} utilized matrix factorization and the log-barrier function to reduce the complexity of sparse channel estimation for one-bit receivers\cite{ref:rwlc3}.
Masood \textit{et al.} applied low-rank matrix sensing problem and low-rank matrix completion for channel estimation, exploiting channel sparsity to reduce the computational complexity\cite{ref:rwlc4}.
Ruan \textit{et al.} proposed a GAN-based beamspace channel estimation algorithm that reduces training complexity by leveraging Wasserstein GAN with simplified learned AMP (LAMP)\cite{ref:LAMP-GM}.
The work in \cite{ref:rwlc5} employed subspace-based channel estimation and simplified the process using a low-dimensional least squares problem.
Zhu \textit{et al.} applied semidefinite programming and alternating direction method of multipliers (ADMM) to simplify the original atomic norm minimization problem of channel estimation, leveraging Hankel-Toeplitz structures\cite{ref:rwlc6}.
Fesl \textit{et al.} devised a lightweight CNN within a diffusion-based generative prior to achieve low-complexity MIMO channel estimation\cite{ref:rwlc7}.
In \cite{ref:beaches}, a soft-thresholding-based beamspace channel estimation algorithm was devised, which adaptively denoises the noisy channel using SURE with a limited number of iterations.
However, these algorithms still involve iterative search, optimization, or learning phases, leaving room for further improvement in simplicity of implementation and computational efficiency.

\subsection{Contributions}
Since mmWave channels exhibit rapid fluctuations\cite{ref:mmwavetextbook}, fast and computationally efficient channel estimation becomes essential.
However, the complexity of existing estimation algorithms typically scales with the number of antennas and further increases with iterative search or learning, which poses significant challenges in terms of latency and computing resource requirements in large-scale mmWave MIMO systems.
To address this challenge, we propose a novel beamspace channel denoiser for mmWave massive MIMO systems that achieves low computational complexity by leveraging Bayesian binary hypothesis testing based on a Bernoulli-complex Gaussian distribution, which effectively captures the sparsity of beamspace channels. Moreover, the proposed method requires neither iterative optimization nor a learning-based approach.
The detailed contributions of this work are summarized as follows:
\begin{itemize}
    \item We reformulate beamspace channel denoising as an element-wise binary hypothesis testing problem under Bernoulli-complex Gaussian distribution, enabling the identification of signal-relevant components through a closed-form likelihood ratio test. The proposed method denoises the noisy channel by zeroing out the magnitudes of elements that are classified as noise, using a hard-thresholding method. This formulation allows the proposed method to operate without iterative optimization or learning-based training.
    \item We provide theoretical expressions for the detection and false alarm probabilities of the proposed detector, and further derive closed-form bounds on the estimation error. These results demonstrate that the estimation performance improves as the beamspace sparsity increases, highlighting the suitability of the method under sparse channel conditions.
    \item The proposed algorithm is designed to operate with linear computational complexity $\mathcal{O}(M)$ for $M$-dimensional channel vector. All major components of the proposed method, including parameter estimation and thresholding, are implemented without matrix inversion or iterative procedures, ensuring suitability for real-time and resource-constrained systems.
    \item We validate the analytical results through extensive simulations, which confirm the consistency between theoretical predictions and empirical behavior. The proposed method is evaluated under various settings, including changes in the detection threshold, parameter estimation errors, and both synthetic and realistic channel models. Simulation results show that the proposed method achieves estimation accuracy comparable to other baselines, with negligible degradation compared to computationally intensive deep-learning-based methods.
\end{itemize}

\subsection{Paper Outline}
The rest of this paper is organized as follows. 
In Section~\ref{sec:systemmodel}, the system model considered in this paper is described.
In Section~\ref{sec:htbbce}, the proposed channel denoiser is described in detail.
In Section~\ref{sec:simulationresults}, the numerical results and the performance analysis based on them are described.
Then, Section~\ref{sec:conclusion} concludes the paper.
\subsubsection*{Notation} 
Matrices and vectors are denoted by boldface uppercase and lowercase letters, respectively.
All vectors are column vectors. 
For matrix $\mathbf{A}$, its $i$th column vector is $\mathbf{a}_i$, and its $j$th element is $A_{ij}$.
$(\cdot)^T$ means transpose, $(\cdot)^H$ denotes conjugate transpose, and $\mathbf{I}_N$ is $N\times N$ identity matrix.
$\mathbf{0}^{N\times M}$ is $N\times M$ zero matrix.
$\|\cdot\|$ represents the $\ell_2$ norm, and $\|\mathbf{a}\|_0$ is the $\ell_0$ norm, i.e., the number of nonzero entries in $\mathbf{a}$.
Also, $|\mathbf{a}|$ means the element-wise absolute-valued vector for $\mathbf{a}$.
$\mathrm{Exp}(\lambda)$ denotes exponential distribution with rate parameter $\lambda$, with the probability density function (PDF)
\begin{equation}
    f_\mathrm{Exp}(x, \lambda) = \begin{cases}
        \lambda e^{-\lambda x} & x\geq 0, \\
        0 & x<0.
    \end{cases}
\end{equation}
$\overset{D}{\longrightarrow}$ and $\overset{\textnormal{a.s.}}{\longrightarrow}$ denote convergence in distribution and convergence almost surely, respectively.
Also, $\underset{\mathcal{H}_0}{\overset{\mathcal{H}_1}{\gtrless}}$ denotes a binary decision operator for arbitrary binary hypotheses $\mathcal{H}_0$ and $\mathcal{H}_1$ in which the decision is $\mathcal{H}_1$ if the left-hand side is greater than the right-hand side, and the decision is $\mathcal{H}_0$ otherwise.

\section{System Model}\label{sec:systemmodel} 
We consider a single-cell massive MU-MIMO mmWave uplink system. 
The base station (BS) is equipped with $M$-element uniform linear array (ULA), and $K$ single-antenna users share the same frequency and time resources.
The received baseband signal $\mathbf{y}\in\mathbb{C}^M$ is given by
\begin{equation}
    \mathbf{y}=\mathbf{Hx}+\mathbf{n},
\end{equation}
where $\mathbf{H}\in\mathbb{C}^{M\times K}$ is the channel matrix, $\mathbf{x}\in\mathbb{C}^{K}$ is the transmitted signal, and $\mathbf{n}\sim \mathcal{CN}(0, N_0\mathbf{I}_M)$ is the zero-mean additive white Gaussian noise (AWGN) vector.
The frequency-flat channel vector of user $k$ in mmWave systems is given by \cite{ref:channel}
\begin{equation}
    \mathbf{h}_k=\sum_{\ell=1}^L g_{k,\ell} \mathbf{a}(\phi_{k,\ell}),
\end{equation}
where $L$ is the number of propagation paths including a potential line-of-sight (LoS) path, $g_{k,\ell}$ is the complex channel gain of the $\ell$-th dominant propagation path, and $\mathbf{a}(\phi_{k,\ell})\in\mathbb{C}^M$ is steering vector which is determined by the incident angle of the $\ell$-th path $\phi_{k,\ell}$.
The steering vector corresponding to the incident angle $\phi$ can be expressed as
\begin{equation}
    \mathbf{a}(\phi)=[1, e^{-j2\pi\phi},\cdots,e^{-j2(M-1)\pi\phi}]^T.
\end{equation}
In the following, we remove the user index for notational convenience.
Without loss of generality, we assume that the variance of pilot and data symbols, which is denoted by $P_x$, is 1.
With orthonormal uplink pilots, the BS can obtain a noisy channel observation via dedicated training phase for each UE, which is presented as
\begin{equation}
    \mathbf{h'}=\mathbf{h}+\mathbf{e},
\end{equation}
where $\mathbf{e}\sim\mathcal{CN}(0, E_0\mathbf{I}_M)$ is the channel estimation error. Since we assumed that the variance of the pilot is 1, $E_0$ is equivalent to $N_0$.
To perform beamspace processing, the observed channel is transformed from the antenna domain into beamspace, which can be expressed as
\begin{equation}
    \bar{\mathbf{h}}^\prime = \mathbf{Fh}^\prime = \mathbf{Fh} + \mathbf{Fe}=\bar{\mathbf{h}}+\bar{\mathbf{e}},
\end{equation}
where $\bar{\mathbf{h}}$ and $\bar{\mathbf{e}}$ are the ground-truth channel and channel estimation error in beamspace domain, respectively, and $\mathbf{F}\in\mathbb{C}^{M\times M}$ is the normalized discrete Fourier transform matrix implying $\mathbf{F}^H=\mathbf{F}^{-1}$.
Note that $\bar{\mathbf{e}}$ is still AWGN, and $\|\mathbf{h}\|=\|\bar{\mathbf{h}}\|$.
Due to the limited scattering in mmWave propagation environments, noiseless channel vector $\bar{\mathbf{h}}$ is approximately sparse.
Using this assumption, the sparsity of the beamspace channel can be simplified with Bernoulli-complex Gaussian distribution\cite{ref:rwcs4, ref:rwcs5, ref:LAMP-GM, ref:MAD}, which implies that the PDF of noiseless beamspace channel vector element is determined as
\begin{equation}
        f_{\bar{\mathbf{h}}}(\bar{h}_m) \triangleq~ q\frac{1}{\pi \|\bar{\mathbf{h}}\|^2/q}
        \exp \left( -\frac{|\bar{h}_m|^2}{\|\bar{\mathbf{h}}\|^2/q} \right) 
        + (1-q) \delta(\bar{h}_m),
\end{equation}
where $q=\|\bar{\mathbf{h}}\|_0/M$ is activity rate of the channel and $\delta(\cdot)$ is the Dirac-delta distribution.
To provide intuition, the activity rate $q$ indicates the fraction of active beams in the beamspace channel vector. 
For example, in a massive MIMO system with $M=128$ antennas, if only 20 elements in the noiseless beamspace channel vector $\bar{\mathbf{h}}$ carry signal energy and other elements have near-zero magnitudes, then $q=20/128=0.15625$.
This reflects the inherent sparsity of mmWave channels in beamspace caused by limited scattering, where only a small subset of angular directions contains significant multipath components.
Based on this modeling, the noisy observation of the mmWave beamspace channel can be derived as
\begin{equation}\label{eq:pdf}
    \begin{aligned}
    f_{\bar{\mathbf{h}^\prime}}(\bar{h}_m^\prime) \triangleq~
    & q \frac{1}{\pi(E_0+\|\bar{\mathbf{h}}\|^2/q)} \exp\left(-\frac{|\bar{h}_m^\prime|^2}{E_0+\|\bar{\mathbf{h}}\|^2/q}\right) \\
    & + (1-q) \frac{1}{\pi E_0} \exp\left(-\frac{|\bar{h}_m^\prime|^2}{E_0} \right).
    \end{aligned}
\end{equation}
The PDF in \eqref{eq:pdf} models the noisy observation by assigning each beam a probability $q$ of being signal-bearing and $(1-q)$ of being noise-only.

\begin{remark}
    We focus on fully digital architectures, where the full-dimensional beamspace channel observation is directly available at the BS. 
    For hybrid beamforming architectures, the proposed denoiser could be combined with a channel reconstruction stage that recovers the beamspace channel from reduced-dimensional combined measurements. In this case, the overall channel estimation performance and computational complexity would also depend on the reconstruction stage, in addition to the proposed denoiser. A detailed investigation of this extension is left for future work.
\end{remark} 
\begin{remark}
    The considered Bernoulli-complex Gaussian model does not explicitly capture the dependencies among neighboring beamspace components, whereas practical mmWave beamspace channels may exhibit clustered sparsity, which can be interpreted as a block-sparse signal denoising problem with structured priors. Investigating such structured sparsity within the proposed framework is left for future work.
\end{remark}

\section{Hypothesis Testing-Based Beamspace Channel Estimation}\label{sec:htbbce}
To achieve low computational complexity in beamspace channel estimation, we propose a lightweight denoising algorithm with time complexity $\mathcal{O}(M)$ based on element-wise binary hypothesis testing under the system model described in Section~\ref{sec:systemmodel}.
Since the beamspace channel is sparse, it is reasonable to assume that many elements are inactive, i.e., their true values are near-zero, while only a few carry significant signal energy.
This motivates a statistical decision framework that classifies each component as either signal-dominant or noise-only.
Specifically, for each entry of the observed beamspace channel vector, we formulate a binary hypothesis test to determine whether the element corresponds to an active signal-bearing beam or is dominated by noise.
To this end, the considered system model can be rewritten as the following hypothesis test:
\begin{equation}
    \bar{h}^\prime_m=
    \begin{cases}
        \bar{e}_m, & \mathcal{H}_0\\
        \bar{h}_m + \bar{e}_m, & \mathcal{H}_1,
    \end{cases}
\end{equation}
where $\mathcal{H}_0$ corresponds to the case where the element contains only noise and $\mathcal{H}_1$ means the element is mixture of signal and noise.
Also, based on the considered system model, priors of the hypotheses are equivalent to
\begin{equation}
    p(\mathcal{H}_0)=1-q,
\end{equation}
\begin{equation}
    p(\mathcal{H}_1)=q.
\end{equation}
Here, $p(\mathcal{H}_0)$ and $p(\mathcal{H}_1)$ represent the prior probabilities that each beamspace component is noise-dominated and signal-relevant, respectively. These priors are derived from the Bernoulli-complex Gaussian sparsity model of the beamspace channel, i.e., the fraction of active beams. Since mmWave channels are typically sparse, $q$ is small in practice, which justifies the asymmetric prior favoring $\mathcal{H}_0$\cite{ref:mmwavetextbook}.

Based on this testing, Bayesian hypothesis testing can be organized as
\begin{equation}
    \frac{p(\bar{\mathbf{h}}^\prime|\mathcal{H}_1)}{p(\bar{\mathbf{h}}^\prime|\mathcal{H}_0)}
    \underset{\mathcal{H}_0}{\overset{\mathcal{H}_1}{\gtrless}}
    \frac{p(\mathcal{H}_0) C_{10}}{p(\mathcal{H}_1) C_{01}},
\end{equation}
where $C_{ij}$ is the cost for deciding $i$ when the true state is $j$.
For notational brevity, we simplify the cost as $C_{10}/C_{01}=C$.
The costs reflect the relative penalty of false alarm (deciding $\mathcal{H}_1$ under $\mathcal{H}_0$) versus misdetection (deciding $\mathcal{H}_0$ under $\mathcal{H}_1$) in Bayesian decision theory.
In our context, this parameter determines the aggressiveness of signal detection: a higher value of $C$ raises the detection threshold, suppressing noise at the expense of possibly missing weak signal components.
Conversely, a lower value of $C$ leads to more sensitive detection at the cost of increased false alarm probability, resulting in weaker denoising.
Since we consider sparse channels, false alarms may have a more significant impact on the denoising performance than misdetections. Therefore, $C$ should be chosen to be greater than 1, because $C<1$ implies the increment in the false alarm rate.
The selection of the cost $C$ and its impact on the decision rule and channel estimation performance are discussed in detail in Section~\ref{sec:simulationresults}.

Based on the considered Bernoulli-complex Gaussian distribution, the test can be reformulated as
\begin{equation}\label{eq:test}
    \frac{\frac{1}{\pi(\|\bar{\mathbf{h}}\|^2/p(\mathcal{H}_1)+E_0)}\exp(-\frac{|\bar{h}_m^\prime|^2}{\|\bar{\mathbf{h}}\|^2/p(\mathcal{H}_1)+E_0})}{\frac{1}{\pi E_0}\exp(-\frac{|\bar{h}_m^\prime|^2}{E_0})}
    \underset{\mathcal{H}_0}{\overset{H_1}{\gtrless}} 
    \frac{1-p(\mathcal{H}_1)}{p(\mathcal{H}_1)}C,
\end{equation}
which can be reorganized as
\begin{equation}\label{eqn:threshold_f}
    |\bar{h}_m^\prime|^2 
    \underset{\mathcal{H}_0}{\overset{\mathcal{H}_1}{\gtrless}} 
    E_0\left( \frac{p(\mathcal{H}_1)}{\mathrm{SNR}}+1 \right) 
    \ln\left(\left(1 + \frac{\mathrm{SNR}}{p(\mathcal{H}_1)}\right) \frac{1 - p(\mathcal{H}_1)}{p(\mathcal{H}_1)} C \right),
\end{equation}
where $\mathrm{SNR}=\|\bar{\mathbf{h}}\|^2/E_0=\|\mathbf{h}\|^2/E_0$ is the signal-to-noise ratio (SNR).
For the complete derivation, please refer to Appendix \ref{app:derivation}.
This threshold increases with the noise power and signal sparsity, intuitively suppressing components less likely to be active.

However, this test requires three types of prior knowledge that are typically unavailable prior to channel estimation, (i) SNR, (ii) noise power, and (iii) activity rate $q$.
Accordingly, we employ the blind estimators for the noise power and SNR devised in \cite{ref:MAD} which only require $\mathcal{O}(M)$ to estimate both the noise power and SNR, and estimate the activity rate based on the estimated noise power. 
The noise power estimator proposed in \cite{ref:MAD} utilizes the median absolute deviation (MAD) estimator, and its estimate is determined as
\begin{equation}
    \widehat{E}_0 = M\times\frac{\textnormal{median}(|\bar{\mathbf{h}}^\prime|^2)}{\ln2}.
\end{equation}
The median can be found at a complexity of $\mathcal{O}(M)$ time either in both the worst and average cases using median of medians\cite{ref:medianofmedians} or on average using the quickselect algorithm\cite{ref:quickselect}.
Also, it is used in blind SNR estimation, which is given by
\begin{equation}\label{eq:snrest}
    \widehat{\mathrm{SNR}} = \max\left\{ \frac{\|\bar{\mathbf{h}}^\prime\|^2}{\widehat{E}_0} -1 ,0 \right\}.
\end{equation}

Based on the noise power estimate, we propose a low-complexity activity rate estimator. 
To estimate the activity rate $q$, we utilize the method of moments (MoM) with the fourth-order moment of the noisy observation of the channel, which can be presented as
\begin{equation}
\mu_4  \triangleq
    \mathbb{E}[|\bar{\mathbf{h}}^\prime|^4]
     = p(\mathcal{H}_0)\mathbb{E}[|\bar{\mathbf{h}}^\prime|^4|\mathcal{H}_0]+p(\mathcal{H}_1)\mathbb{E}[|\bar{\mathbf{h}}^\prime|^4|\mathcal{H}_1],
\end{equation}
with the law of total expectation.
First, in Bernoulli-complex Gaussian model, elements of $\bar{\mathbf{h}}$ corresponding to $\mathcal{H}_0$ are 0, so they only contain noise in the noisy observation $\bar{\mathbf{h}}^\prime$. Thus,
$\bar{h}_m^\prime|\mathcal{H}_0\sim\mathcal{CN}(0, E_0)$, and its fourth-order moment is given by
\begin{equation}
    \mathbb{E}[|\bar{\mathbf{h}}^\prime|^4|\mathcal{H}_0] = 2E_0^2,
\end{equation}
which is derived from the well-known result of fourth-order moment of zero-mean complex Gaussian. Likewise, the elements of $\bar{\mathbf{h}}$ corresponding to $\mathcal{H}_1$ follow Gaussian distribution, and their noisy observations are also Gaussian, i.e., $\bar{h}_m^\prime|\mathcal{H}_1\sim\mathcal{CN}(0, E_0+\|\bar{\mathbf{h}}\|^2/q)$. Thus, the fourth-order moment in $\mathcal{H}_1$ is determined as
\begin{equation}
    \begin{aligned}
        \mathbb{E}[|\bar{\mathbf{h}}^\prime|^4|\mathcal{H}_1] & = 2 \left(E_0 + \frac{\|\bar{\mathbf{h}}\|^2}{q} \right)^2 \\
        & = 2\left(E_0^2 + 2E_0\frac{\|\bar{\mathbf{h}}\|^2}{q} + \left(\frac{\|\bar{\mathbf{h}}\|^2}{q}\right)^2 \right) \\
        & = 2E_0^2 + 4E_0\frac{\|\bar{\mathbf{h}}\|^2}{q} + 2\left(\frac{\|\bar{\mathbf{h}}\|^2}{q}\right)^2.
    \end{aligned}
\end{equation}
Accordingly, it can be organized as
\begin{equation}
    \begin{aligned}
        &\mathbb{E}[|\bar{\mathbf{h}}^\prime|^4]\\
        & = p(\mathcal{H}_0)\mathbb{E}[|\bar{\mathbf{h}}^\prime|^4|\mathcal{H}_0]+p(\mathcal{H}_1)\mathbb{E}[|\bar{\mathbf{h}}^\prime|^4|\mathcal{H}_1] \\
        & = (1-q)\cdot 2E_0^2 + q\cdot \left( 2E_0^2 + 4E_0\frac{\|\bar{\mathbf{h}}\|^2}{q} + 2\left(\frac{\|\bar{\mathbf{h}}\|^2}{q}\right)^2 \right) \\
        & = 2(1-q)E_0^2 + 2qE_0^2 + 4E_0\|\bar{\mathbf{h}}\|^2 + 2\frac{\|\bar{\mathbf{h}}\|^4}{q} \\
        & = 2E_0^2 + 4E_0\|\bar{\mathbf{h}}\|^2 + 2\frac{\|\bar{\mathbf{h}}\|^4}{q} \\
        & = 2E_0^2 + 4E_0^2\frac{\|\bar{\mathbf{h}}\|^2}{E_0} + 2E_0^2 \frac{\|\bar{\mathbf{h}}\|^4}{qE_0^2} \\
        & = 2E_0^2 + 4E_0^2\mathrm{SNR} + \frac{2E_0^2(\mathrm{SNR})^2}{q}.
    \end{aligned}
\end{equation}
From the expression of $\mu_4$, we can isolate $q$ as follows:
\begin{equation}\label{eq:activity_rate_est_mid}
    q = \frac{2\|\mathbf{h}\|^4}{\mu_4-2E_0^2 - 4E_0\|\mathbf{h}\|^2}
    = \dfrac{2(\mathrm{SNR})^2}{\dfrac{\mu_4}{E_0^2} - 2 - 4\mathrm{SNR}}.
\end{equation}
Also, since true value of fourth-order moment is not generally attainable, we use sample moment 
\begin{equation}
    \hat{\mu}_4 = \frac{1}{M}\sum_{m=1}^M|\bar{h}_m^\prime|^4.
\end{equation}
Based on this relation, the estimated value of activity rate $q$ can be derived in terms of $\hat{\mu}_4$, $\widehat{E}_0$, and $\widehat{\mathrm{SNR}}$.
By the definition of the activity rate, it represents the fraction of active elements among $M$ total components. 
Hence, $q$ must take a value from the set $\left\{ \frac{m}{M} | m=1,2,...,M \right\}$.
With this constraint, the estimate for $q$ is determined by quantization as
\begin{equation}\label{eq:activity_rate_est}
    \hat{q}=\argmin_{q^\prime\in \{\frac{m}{M} | m=1,...,M\} }
    \left| \frac{2(\widehat{\mathrm{SNR}})^2}{\dfrac{\hat{\mu}_4}{E_0^2}-2-4\widehat{\mathrm{SNR}}} - q^\prime \right|.
\end{equation}

\noindent Note that this estimator also requires only $\mathcal{O}(M)$.
The performance of the estimators and the influence of the estimation error are discussed in Section~\ref{sec:simulationresults}.
Accordingly, the testing can be rewritten with the estimated values as follows:
\begin{equation}
    |\bar{h}_m|^2 
    \underset{\mathcal{H}_0}{\overset{\mathcal{H}_1}{\gtrless}} 
    \widehat{E}_0\left( \frac{\hat{q}}{\widehat{\mathrm{SNR}}}+1 \right)
    \ln\left(\left(1 + \frac{\widehat{\mathrm{SNR}}}{\hat{q}}\right) \frac{1 - \hat{q}}{\hat{q}} C \right).
\end{equation}

After classifying the elements as either noise-only or signal-dominant, we denoise the channel vector by setting the noise-only components to zero, as given by
\begin{equation}
    \bar{h}_m^\star =
    \begin{cases}
        0, & \textnormal{if}~ \mathcal{H}_0\\
        \bar{h}^\prime_m, & \textnormal{if}~ \mathcal{H}_1. 
    \end{cases}
\end{equation}
Note that this procedure can be interpreted as a hard-thresholding operation applied to $\bar{\mathbf{h}}^\prime$, with the threshold given by
$\widehat{E_0}\left( {\hat{q}}/{\widehat{\mathrm{SNR}}}+1 \right) \ln\left((1 + \widehat{\mathrm{SNR}}/{\hat{q}} ) \frac{1 - \hat{q}}{\hat{q}} C \right).$

Meanwhile, the performance of the proposed algorithm can be characterized analytically, as it is grounded in a statistical framework.
In particular, we evaluate the detection probability ($P_D$) and false alarm probability ($P_{FA}$) within the binary hypothesis testing paradigm.
The detection probability $P_D$ denotes the probability of correctly declaring a signal-relevant element, whereas the false alarm probability $P_{FA}$ refers to the probability of incorrectly declaring an element as signal-relevant when the true state is noise-relevant, i.e., a type-I error.
By the Gaussian properties of the channel and noise model, they have following closed-form expressions:
\begin{equation}\label{eq:P_D}
    P_D=\exp\left( -\frac{\tau}{E_0+\|\bar{\mathbf{h}}\|^2/q}  \right),
\end{equation}
\begin{equation}
    P_{FA} = \exp\left( -\frac{\tau}{E_0} \right),
\end{equation}
where $\tau$ is the detection threshold.
The detailed derivation is conducted in Appendix~\ref{app:pfapd}.
The expression shows that an increase in noise power results in a higher false alarm probability, whereas greater sparsity leads to a higher detection probability.
These probabilities directly affect the accuracy of the channel estimation.
In particular, they influence the mean-squared error (MSE) of the estimated channel, which can be approximated in closed-form.
Theorem~\ref{thm:bound} provides an analytical approximation of the MSE in terms of the detection and false alarm probabilities.
\begin{theorem}[Expected MSE]\label{thm:bound}
    Given the proposed channel estimation algorithm and the aforementioned system model, the MSE of the estimated channel is approximately expressed as:
    \begin{equation}
        \mathrm{MSE} \approx q(P_DE_0+(1-P_D)\|\bar{\mathbf{h}}\|^2)+(1-q)P_{FA}E_0,
    \end{equation}
    where MSE is defined as
    \begin{equation}
        \mathrm{MSE}=\frac{1}{M} \sum_{m=1}^M \mathbb{E}[|\bar{h}_m^\star-\bar{h}_m|^2].
    \end{equation}
\end{theorem}
\begin{proof}
    Please refer to Appendix~\ref{app:bound}.
\end{proof}
\noindent 
In the asymptotic regime where the element-wise signal detection becomes nearly perfect, i.e., $P_D\to1$ and $P_{FA}\to 0$, the approximated term reduces to a simple form that isolates the contribution of noise on the active components. This limiting behavior is summarized in Corollary~\ref{cor:inf}.
\begin{corollary}\label{cor:inf}
    Consider the approximation of the expected MSE in Theorem 1. In the limiting case where the detection probability $P_D\to1$ and the false alarm probability $P_{FA}\to0$, it converges to:
    \begin{equation}
        \mathrm{MSE}\to qE_0.
    \end{equation}
\end{corollary}
\begin{proof}
    Please refer to Appendix~\ref{app:proof_cor}.
\end{proof}
\noindent This expression reflects the residual distortion due solely to additive noise on the correctly identified signal support.
Notably, this approximation also illustrates that the estimation error decreases as the beamspace channel has more sparsity. This scaling behavior is formalized below.
\begin{corollary}\label{cor:prop}
    Under fixed SNR, if $P_D\gg P_{FA}$ and $P_D\gg 1-P_D$, the MSE of the proposed channel estimation algorithm satisfies the following scaling behavior:
    \begin{equation}
        \mathrm{MSE}\approx\beta q \propto q,
    \end{equation}
    where $\beta$ is a positive constant.
\end{corollary}
\begin{proof}
    Please refer to Appendix~\ref{app:proof_cor2}.
\end{proof}

\noindent
The performance of the proposed activity rate estimator also can be theoretically analyzed because it is based on the statistical function. 
Theorem~\ref{thm:bound_act} shows the bounds on the variance of MoM-based estimator.
\begin{theorem}\label{thm:bound_act}
    Given the proposed activity rate estimator in~\eqref{eq:activity_rate_est}, sufficiently large $M$, and assuming $\widehat{E}_0\approx E_0$ and $\widehat{\mathrm{SNR}}\approx\mathrm{SNR}$, there exist constants $c_1,c_2>0$ which satisfy
    \begin{equation}
        \frac{c_1 q^2}{M}\leq\mathrm{Var}(\hat{q})\leq \frac{c_2 q^2}{M},
    \end{equation}
    and thus $\mathrm{Var}(\hat{q}) \propto q^2.$
\end{theorem}
\begin{proof}
    Please refer to Appendix~\ref{app:proof_thm2}.
\end{proof}
\noindent This implies that the estimator exhibits increased stability under higher sparsity levels, i.e., smaller $q$.

In summary, the proposed method performs signal denoising with five main steps, which are shown in Algorithm~\ref{alg:alg1}: 1) blind noise power estimation, 2) blind SNR estimation, 3) blind activity rate estimation, 4) threshold calculation, and 5) hard thresholding.
In terms of computational complexity, blind estimators for noise power, SNR, and blind activity rate estimation require $\mathcal{O}(M)$ for each. 
Also, since threshold calculation is derived as a closed-form, it only requires $\mathcal{O}(1)$.
The final step, hard thresholding, requires $\mathcal{O}(M)$ to scan all elements of beamspace channel vector and perform detection. 
Overall, the estimator only requires $\mathcal{O}(M)$, which represents extremely reduced complexity, which is discussed in detail compared to other channel estimation algorithms in Section~\ref{sec:simulationresults}.

\begin{algorithm}[t]
\caption{Hypothesis Testing-Based Beamspace Channel Estimation}\label{alg:alg1}
\SetAlgoLined
\textbf{input:} Noisy beamspace channel vector $\bar{\mathbf{h}}^\prime$, testing parameter $C$\\
\textbf{initialization:} estimated beamspace channel vector $\bar{\mathbf{h}}^\star=\mathbf{0}^{M\times 1}$\\
Calculate the noise power estimate $\widehat{E}_0$ \\
Compute $\widehat{\mathrm{SNR}}$ and $\hat{q}$ using $\widehat{E}_0$\\
Obtain the threshold with $\widehat{E}_0,\widehat{\mathrm{SNR}},\hat{q}$ as
\begin{equation}
    \tau = \widehat{E}_0\left( \frac{\hat{q}}{\widehat{\mathrm{SNR}}}+1 \right) \ln \left(\left( 1+\frac{\widehat{\mathrm{SNR}}}{\hat{q}} \right) \frac{1-\hat{q}}{\hat{q}}C \right)
\end{equation}\\
\For{$m=1:M$}{
    \If{$|\bar{h}^\prime_m|^2\geq \tau$}{
    $\bar{h}_m^\star = \bar{h}_m^\prime$.
    }
}
\Return{\textnormal{beamspace channel estimate} $\bar{\mathbf{h}}^\star$}    
\end{algorithm}


\section{Simulation Results}\label{sec:simulationresults}
In this section, we provide numerical results of the simulation to demonstrate the competence of the proposed algorithm.
First, we describe the simulation setup. Then, we present and analyze the numerical results that illustrate the performance of the proposed method.
\subsection{Simulation Configuration}
To assess the performance of the proposed algorithm, we consider a BS with $128$-element ULA employing all-digital beamforming, and 8 users equipped with single antennas in the cell.
We assume the antenna spacing of BS is half-wavelength, and the antenna of each user is omnidirectional.
For the channel models, we consider both synthetic data following the distribution in \eqref{eq:pdf} and realistic channel models.
The realistic channel models are generated for both LoS and non-LoS (NLoS) scenarios based on the QuaDRiGa mmMAGIC UMi\cite{ref:quadriga} with a carrier frequency of 50 GHz.
The generated channel models contain many near-zero elements but rarely include exact zeros, because the signals in realistic channel models arrive from arbitrary directions and do not align well with DFT bins, resulting in power leakage\cite{ref:leakage}.
Therefore, directly defining the activity rate using the $\ell_0$-norm is infeasible. 
To quantify the beamspace sparsity in a realistic setting where channel coefficients are rarely exactly zero, we define the activity rate based on an energy thresholding criterion.
Let $|\bar{\mathbf{h}}^\textnormal{sort}|^2=[|\bar{h}^\textnormal{sort}_1|^2 ,\cdots, |\bar{h}^\textnormal{sort}_M|^2]^T\in\mathbb{R}^M$ denote the beamspace element-wise channel power vector which is sorted in descending order, i.e., $|\bar{h}^\textnormal{sort}_1|^2\geq\cdots\geq |\bar{h}^\textnormal{sort}_M|^2$.
Then, for a given energy threshold $\eta$, the activity rate is defined as $q_\textnormal{energy}=A_\eta/M$, where the number of active beams $A_\eta$ is
\begin{equation}
    A_\eta = \argmin_a \left\{ \sum_{m=1}^a |\bar{h}_m^\textnormal{sort}|^2 \geq \eta\cdot\|\bar{\mathbf{h}}\|^2 \right\}.
\end{equation} 
This means that the beam components that, even when combined among themselves, contribute less than $(1-\eta)$ of the total channel power are classified as inactive.
For example, if $\eta=0.9$, the beams that constitute 90\% of the total power in the noiseless beamspace channel are considered active, while the remaining smallest beams that represent the residual 10\% are regarded as inactive.
In this paper, we consider $\eta=0.99$, to neglect the 1\% power leakage.\footnote{The energy threshold $\eta$ can also be set to other values, e.g., 0.95, 0.999, or other appropriate values to regard near-zero elements as inactive.}
This reflects the effective sparsity of the channel by capturing how concentrated the energy is among a few dominant beams.
This classification metric is also used to compute $P_D$ and $P_{FA}$.
The cost parameter $C$ used in the hypothesis testing is set to 5, which shows a robust estimation performance. A detailed empirical justification for this choice is presented in a subsequent subsection.

For the simulations, we conducted 10000 Monte Carlo trials.
To conduct performance evaluation of the proposed channel estimation algorithm, we compare its estimation accuracy with baseline channel estimation algorithms.
For the baseline, we consider least-squares (LS) estimation, which is known as a representative and basic pilot-based channel estimation method, and mmWave beamspace channel estimation algorithms that exploit sparsity: BEACHES\cite{ref:beaches} as a low-complexity denoiser, LAMP-GM\cite{ref:LAMP-GM} with the optimal number of Gaussian mixtures as a supervised learning-based method, and LDGEC-SURE\cite{ref:ldgec} as an unsupervised learning-based estimator.
Also, as a reference, the results for ``Perfect detection'' use the same noise-zeroing-based approach for denoising, but it is based on perfect detection of nonzero beamspace channel vector elements. Note that this approach is equivalent to the theoretically ideal case of the proposed algorithm which ignores the estimation errors of activity rate, noise power, and SNR. 

To assess the performance of the proposed method in the application based on channel estimation, we evaluate post-equalization uncoded bit error rate (BER) for 16-quadrature amplitude modulation (QAM).
For this analysis, we apply linear minimum mean-squared error (LMMSE) equalizer. The equalization matrix is presented as
\begin{equation}
    \mathbf{W}=\hat{\mathbf{H}}(\hat{\mathbf{H}}^H\hat{\mathbf{H}}+\frac{1}{\mathrm{SNR}}\mathbf{I}_K)^{-1}.
\end{equation}
Here, we add ``Perfect CSI'', which corresponds to the perfect channel estimation, for baseline to compare the performance of the proposed method and the other baselines with the optimal case.
To investigate the influence of the noise power estimation error, we define relative noise power error as
\begin{equation}
    \widetilde{E}_0=\frac{\widehat{E}_0-E_0}{E_0}.
\end{equation}
\subsection{Complexity Analysis}

\begin{table}[]
    \centering
    \caption{Computational Complexity of Denoisers}
    \begin{tabular}{|c|c|c|c|}
    \hline
    \textbf{Algorithm} & \textbf{Complexity} & \textbf{Runtime }& \textbf{Learning-based} \\ \hline \hline
    BEACHES       &   $\mathcal{O}(M\log M)$  & 2.30$\times$ & No \\ \hline
    LAMP-GM       &   $\mathcal{O}(TM^2)$  & 762.30$\times$ & Yes  \\ \hline
    LDGEC-SURE    &   $\mathcal{O}(TM^3)$  & 1452.18$\times$ & Yes    \\ \hline
    Proposed      &   $\mathcal{O}(M)$    & 1.00$\times$ & No    \\ \hline
    \end{tabular}
    \label{tab:complexity}
\end{table}
\begin{figure*}[ht]
    \centering
    \subfloat[]{\includegraphics[width=0.33\linewidth]{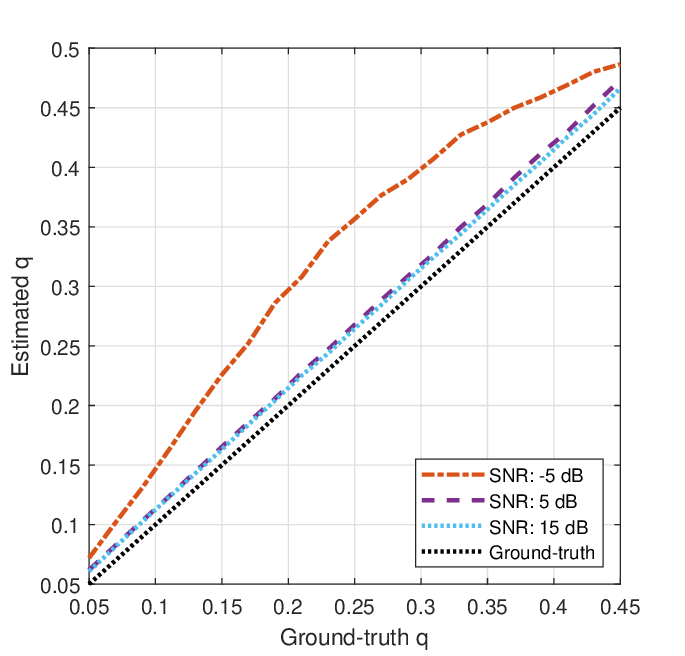}}
    \subfloat[]{\includegraphics[width=0.33\linewidth]{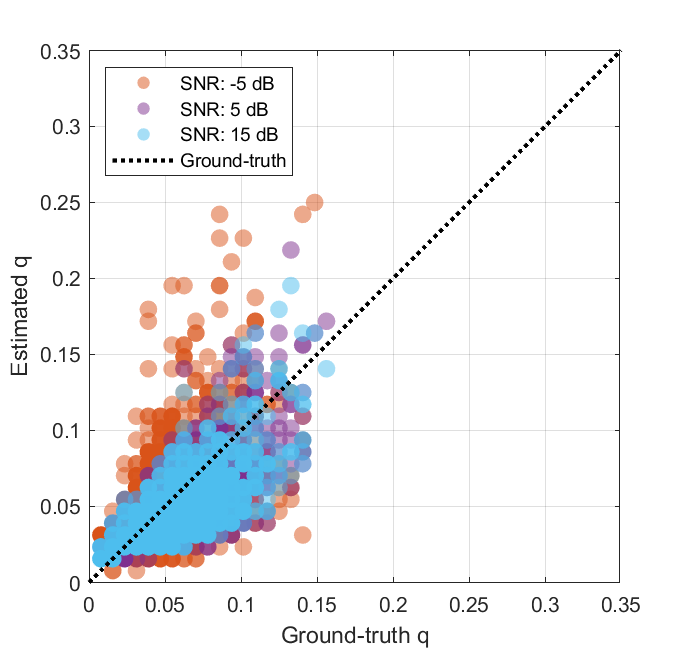}}
    \subfloat[]{\includegraphics[width=0.33\linewidth]{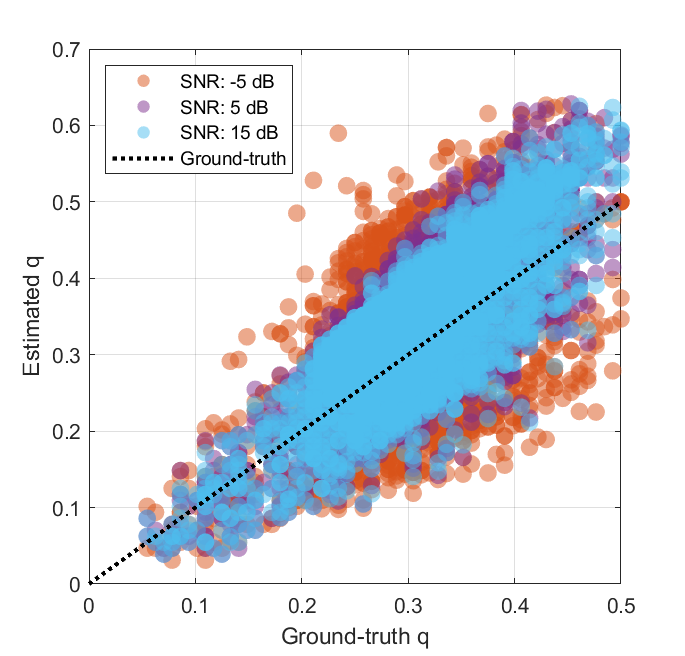}}
    \caption{Estimated activity rate depending on SNR. (a) synthetic channel (b) LoS (c) NLoS.}
    \label{fig:activity}
\end{figure*}

We now compare the computational complexity of the proposed method with that of the baseline estimators.
For a fair comparison, the computational complexity of each algorithm is evaluated from the point at which the channel observation of a single user is transformed into the beamspace domain to the point immediately before it is transformed back into the antenna domain.
Table~\ref{tab:complexity} shows the complexity of channel estimation methods in big-O notation and normalized runtime. 
The proposed method exhibits significantly lower complexity, achieving linear time complexity of $\mathcal{O}(M)$.
Among the baselines, the BEACHES algorithm requires an iterative search to find the optimal soft-thresholding parameter that minimizes the SURE\cite{ref:beaches}.
Each trial of the search requires the complexity of $\mathcal{O}(M)$ to calculate SURE for an arbitrary thresholding parameter.
Although the complexity is reduced to $\mathcal{O}(M\log M)$ by sorting, the iterative search of thresholding parameter still remains.
Due to this structure, it has more than twice the runtime complexity of the proposed method even though it shares a similar thresholding approach.
LAMP-GM is based on AMP, and thus it requires iterative updates for parameters and each update entails $\mathcal{O}(M^2)$.
Also, LAMP replaces the iterative updates of AMP with neural network layers\cite{ref:LAMP-GM}.
As a result, its complexity additionally increases linearly with the number of layers $T$, i.e., $\mathcal{O}(TM^2)$, and it still demands training with labeled data.
LDGEC-SURE further increases the computational burden, requiring two matrix inversions per layer over $T$ layers, with each inversion having a complexity of $\mathcal{O}(M^3)$, and includes multiple convolutional operations for DnCNN\cite{ref:ldgec}.

These distinctions emphasize the structural simplicity of the proposed method.
Unlike the baselines, it does not require complicated operations such as matrix inversion, parameter updates, iterative parameter search, and any form of deep learning.
All necessary parameters of binary hypothesis testing are obtained either in a closed-form or within $\mathcal{O}(M)$ complexity, resulting in a significantly shorter runtime than the considered baselines.
In addition to its low computational complexity, the proposed method requires only $\mathcal{O}(M)$ storage complexity for channel vectors and a small number of scalar parameters, without storing large intermediate matrices.
Moreover, as the proposed method is not based on deep learning, it does not require any training stage or labeled data, which further improves its practicality for real-time and computationally constrained wireless systems.

\subsection{Performance Evaluation}

\begin{figure*}[]
    \centering
    \subfloat[]{\includegraphics[width=0.95\columnwidth]{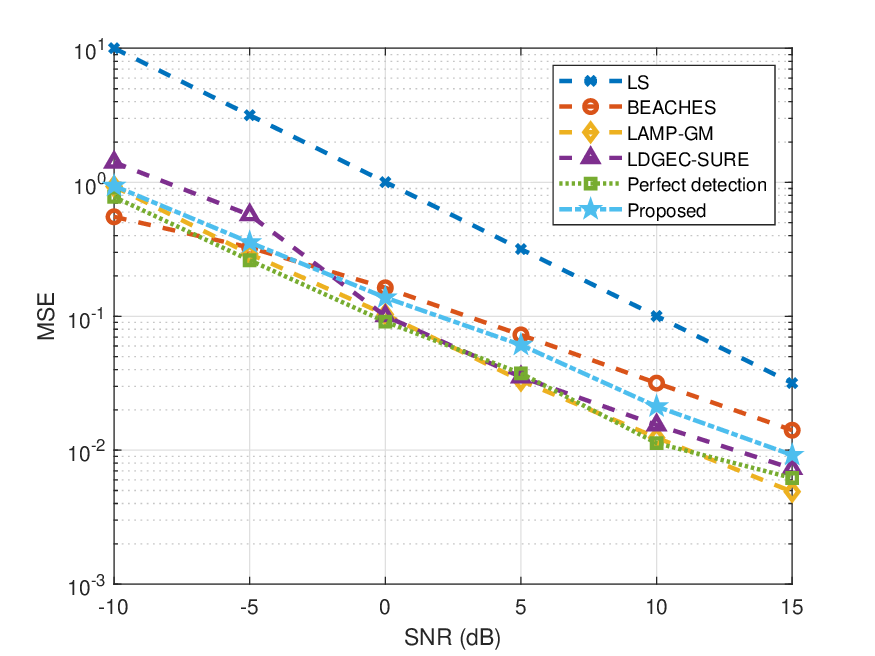}}
    \subfloat[]{\includegraphics[width=0.95\columnwidth]{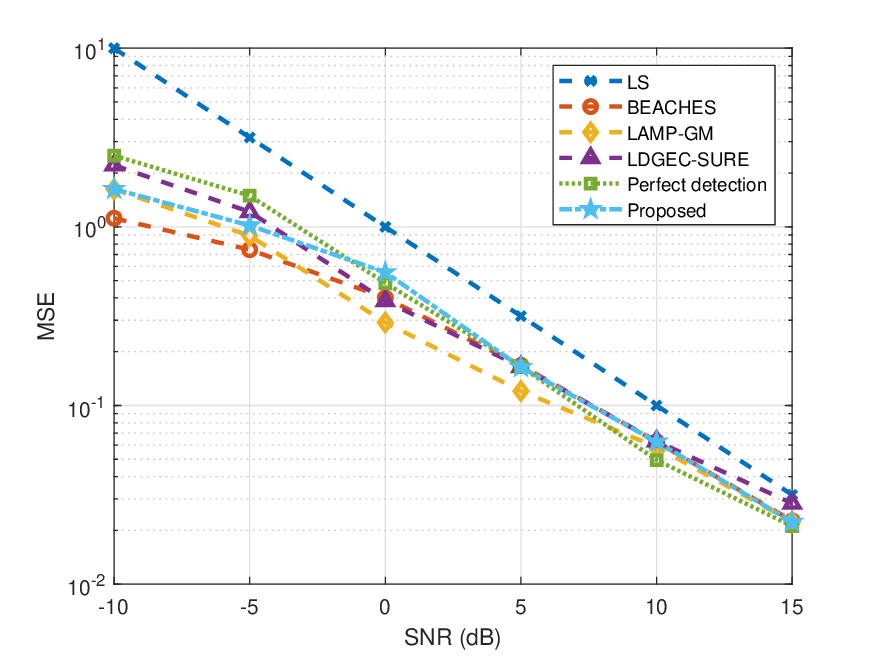}}
    \caption{MSE depending on SNR compared to baselines. (a) LoS (b) NLoS.}
    \label{fig:MSE_baseline}
\end{figure*}

Fig.~\ref{fig:activity} presents the performance of activity rate estimation compared to the ground-truth value for synthetic and realistic channel data.
For the synthetic data, the estimator achieves high accuracy when $\mathrm{SNR}\geq5$ dB, while it tends to slightly overestimate $q$ under low-SNR conditions such as $\mathrm{SNR}=-5$ dB.
As discussed in Theorem~\ref{thm:bound_act}, the estimator shows relatively higher accuracy for a smaller ground-truth $q$.
Meanwhile, for realistic channel data, the estimation accuracy degrades compared to the synthetic case, as the realistic beamspace channels do not strictly follow the Bernoulli-complex Gaussian distribution.
The result shows lower estimation accuracy with larger variance in lower SNR conditions. Also, the LoS case exhibits a smaller estimation variance compared to that of the NLoS case.

\begin{figure*}[ht]
    \centering
    \subfloat[]{\includegraphics[width=0.95\columnwidth]{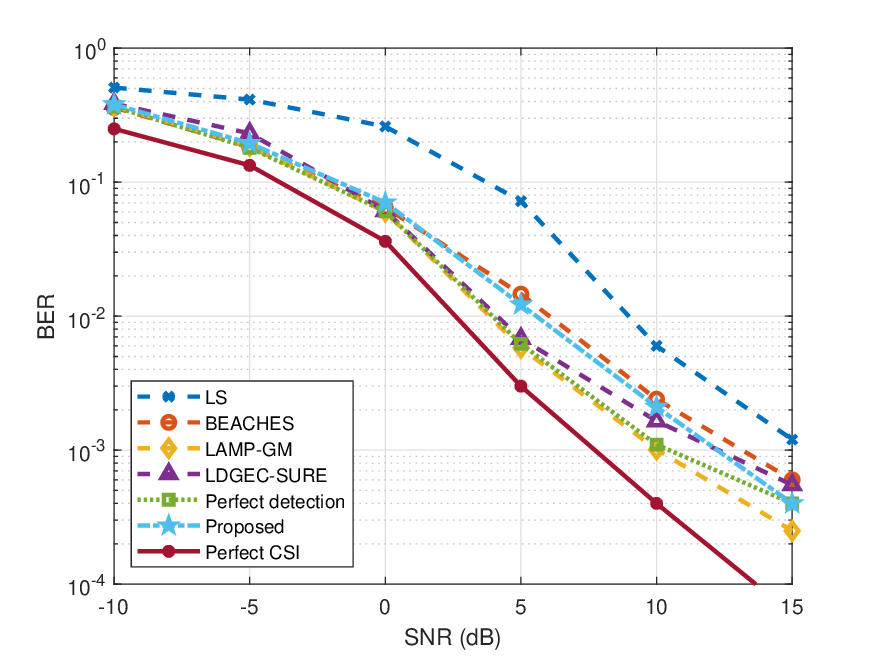}}
    \subfloat[]{\includegraphics[width=0.95\columnwidth]{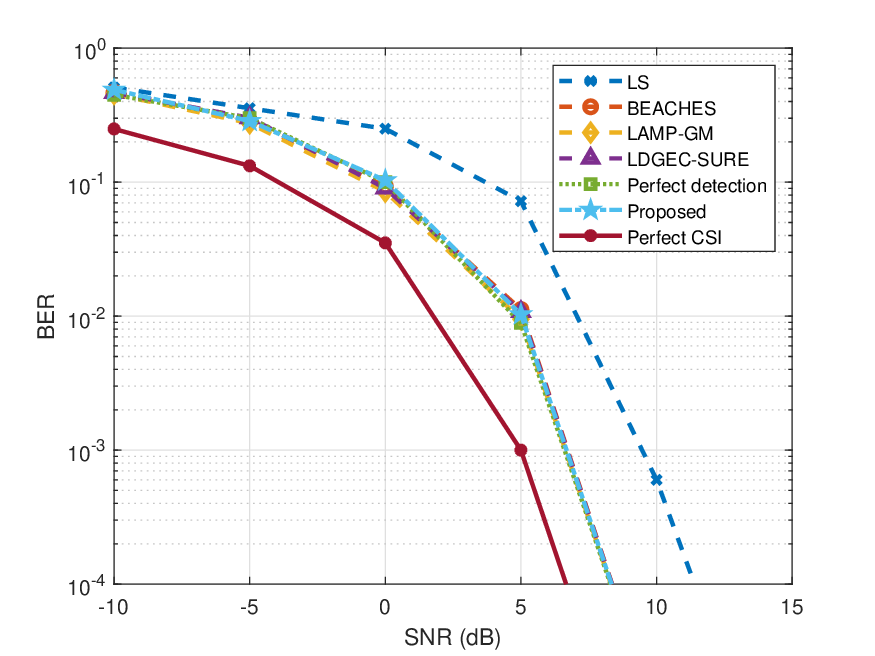}}
    \caption{Post-equalization BER depending on SNR compared to baselines. (a) LoS (b) NLoS.}
    \label{fig:BER}
\end{figure*}

\begin{figure*}
    \centering
    \subfloat[]{\includegraphics[width=0.95\columnwidth]{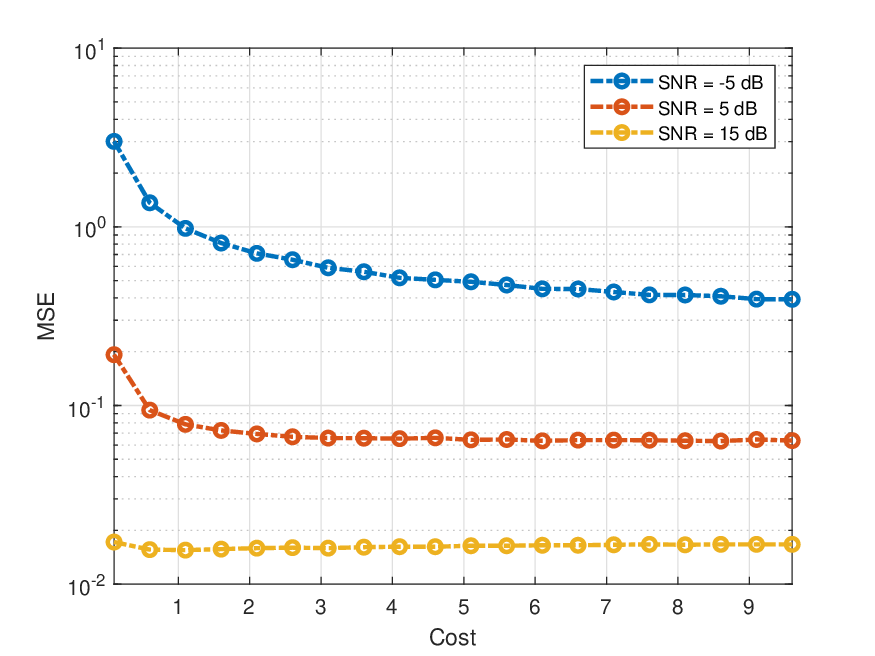}}
    \subfloat[]{\includegraphics[width=0.95\columnwidth]{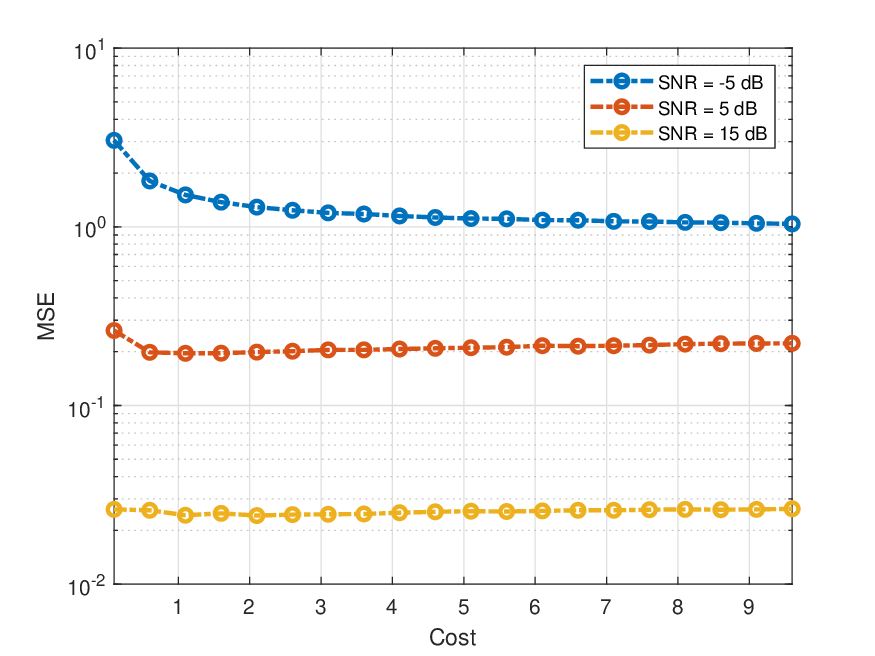}}
    \caption{MSE depending on cost $C$. (a) LoS (b) NLoS.}
    \label{fig:MSE_cost}
\end{figure*}

\begin{figure*}
    \centering
    \subfloat[]{\includegraphics[width=0.95\columnwidth]{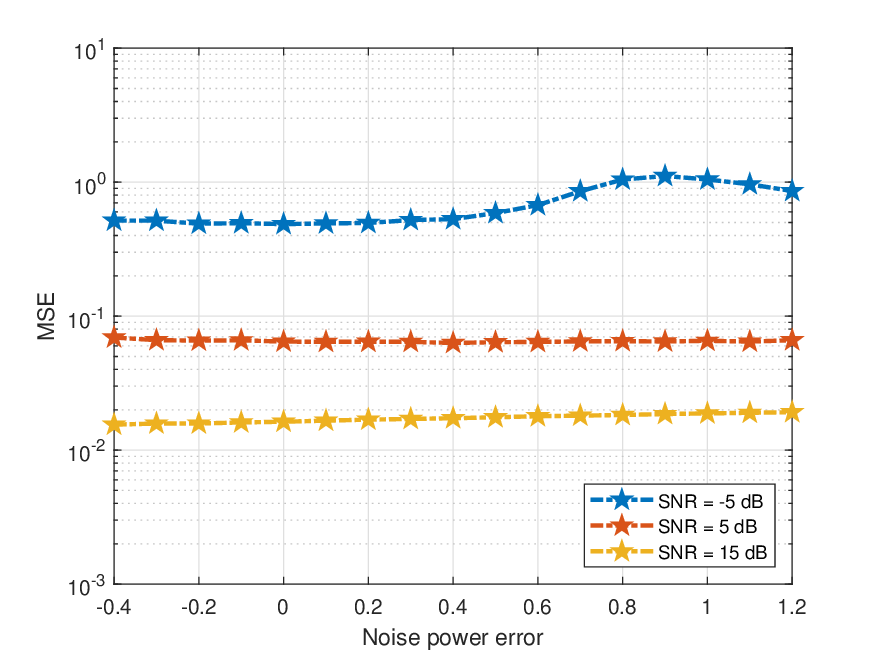}}
    \subfloat[]{\includegraphics[width=0.95\columnwidth]{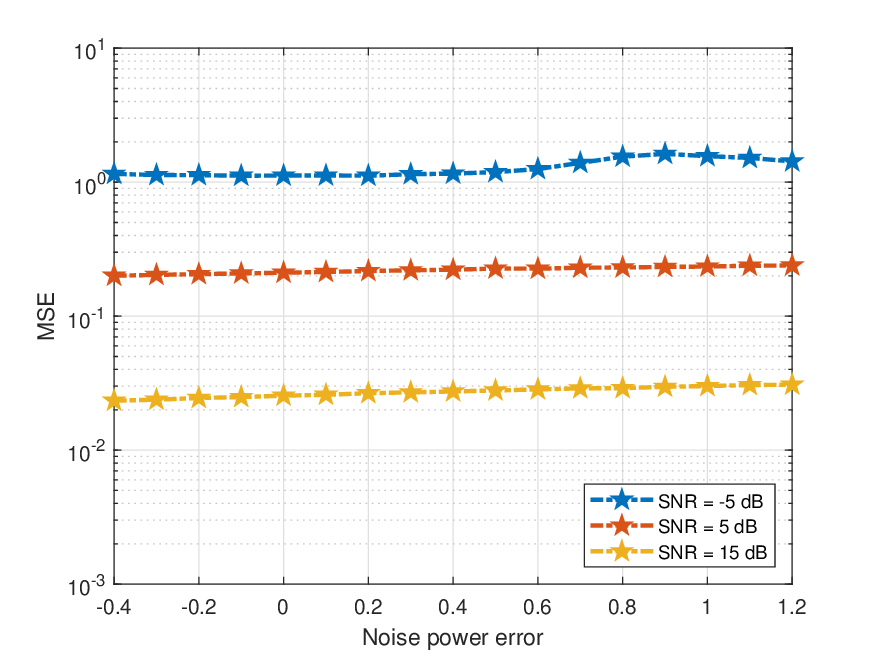}}
    \caption{MSE depending on the noise power estimation error. (a) LoS (b) NLoS.}
    \label{fig:MSE_noiseerr}
\end{figure*}

Fig.~\ref{fig:MSE_baseline} illustrates the MSE of channel estimation compared to other baselines depending on SNR.
Compared to ``Perfect detection'', the proposed algorithm achieves comparable estimation accuracy in both LoS and NLoS scenarios. 
It means that the detection is performed well overall, and the influence of detection errors is not critical in channel estimation performance.
Moreover, in NLoS scenarios, it even outperforms ``Perfect detection'' under low-SNR conditions where SNR is lower than 0 dB.
It is because misdetections may unintentionally suppress more noises by zeroing signals but with strong noises.
The similarity in performance between the proposed method and ``Perfect detection'' suggests that, although parameter estimation errors can affect the channel estimation accuracy, their overall impact on the resulting MSE remains moderate.
Compared to other baselines, the proposed algorithm shows similar MSE.
In particular, in LoS scenarios, the proposed algorithm shows higher estimation accuracy than BEACHES when SNR is higher than 0 dB. 
Accordingly, the proposed algorithm demonstrates channel estimation performance comparable to other baselines, despite its extremely low computational complexity.
When comparing the LoS and NLoS scenarios, the proposed algorithm exhibits more favorable performance in LoS conditions.
This is because the proposed algorithm shows higher estimation accuracy if the channel has more sparsity, as discussed in Corollaries~\ref{cor:inf} and \ref{cor:prop}.

Fig.~\ref{fig:BER} presents the post-equalization uncoded BER depending on SNR.
It reveals that the proposed channel estimation shows comparable performance not only in MSE but also in channel estimation-based application to those of other baselines despite the low complexity, which shows its efficacy.
In both LoS and NLoS scenarios, the difference in BER between the baseline algorithms and the proposed method is almost negligible. Notably, the performance difference is further diminished under NLoS conditions.
This implies that the computation-intensive and deep learning-based baselines do not offer significant performance advantages over the proposed method.

Fig.~\ref{fig:MSE_cost} illustrates the variation of MSE depending on the cost $C$. 
In both LoS and NLoS scenarios, the MSE tends to be higher when $C<1$, except in the NLoS case at the SNR of 15 dB, where the MSE remains nearly constant.
This is because $C<1$ implies $C_{10}<C_{01}$, indicating that the testing is stricter against misdetection but more tolerant toward false alarms, which results in increased residual noise and degraded denoising performance for sparse signals.
In LoS scenarios, increasing $C$ generally leads to a reduction in MSE or preserves it at a similar level, with the most notable improvements occurring when $C\leq3$.
In contrast, under NLoS conditions, increasing $C$ results in a noticeable MSE reduction only in the low-SNR case ($\text{SNR}=-5$ dB), while in higher SNR settings ($\text{SNR}\geq5$ dB), a slight increase in MSE is observed.
Nevertheless, this increase remains negligible.
Overall, selecting $C$ around 5 does not significantly affect estimation performance across various conditions.

Fig.~\ref{fig:MSE_noiseerr} illustrates the MSE depending on the noise power estimation error.
Note that noise power estimation error also propagates to the SNR estimate in \eqref{eq:snrest}, and the activity rate estimate in \eqref{eq:activity_rate_est}.
In both cases, underestimating the noise power---i.e., when the noise power error is negative---does not critically affect the MSE.
Slight overestimation also has a negligible impact on the estimation performance across all SNR regimes.
Even when the noise power is overestimated by nearly a factor of two, the MSE remains stable at 5 dB and 15 dB SNR levels.
In contrast, at low SNRs (e.g., -5 dB), such severe overestimation results in a noticeable increase in the MSE.
Nevertheless, such cases rarely occur in practice, as such severe noise power overestimation is uncommon in low-SNR conditions\cite{ref:MAD}.

Fig.~\ref{fig:ROC} illustrates the receiver operating characteristic (ROC) curve of signal element detection. 
The theoretical predictions and measured results show similar values, which implies that the Bernoulli-complex Gaussian model approximates the beamspace sparse channel well.
As the SNR increases, the ROC curve shifts toward the upper-left corner, indicating a higher probability of detection at a given false alarm probability, which means the significant improvement of detection performance with increasing SNR.
This trend is consistent among the theoretical predictions and measured results, and also in both LoS and NLoS scenarios.
This trend is intuitive: higher noise power makes detection more difficult, while lower noise power implies ease of detection due to larger differences between the magnitudes of noise-relevant and signal-dominant elements.
In LoS scenarios, the measured results exhibit lower estimation performance compared to the theoretical ROC curve.
On the other hand, measured results show similar performance to the theoretical curve in NLoS scenarios.
It is because LoS beamspace channel has a LoS path with dominant gain which has massive magnitude compared to other paths, which makes the approximation to Bernoulli-complex Gaussian prior less efficient.
Comparing the curves, the ROC curve for the LoS case lies further toward the upper-left corner than that of the NLoS case, which reflects better detection capability.
It is because LoS channel has more sparsity than NLoS, and stronger sparsity, i.e., smaller $q$, leads to improved detection probability, which is described in \eqref{eq:P_D}.

\section{Conclusion}\label{sec:conclusion}
\begin{figure*}
    \centering
    \subfloat[]{\includegraphics[width=0.95\columnwidth]{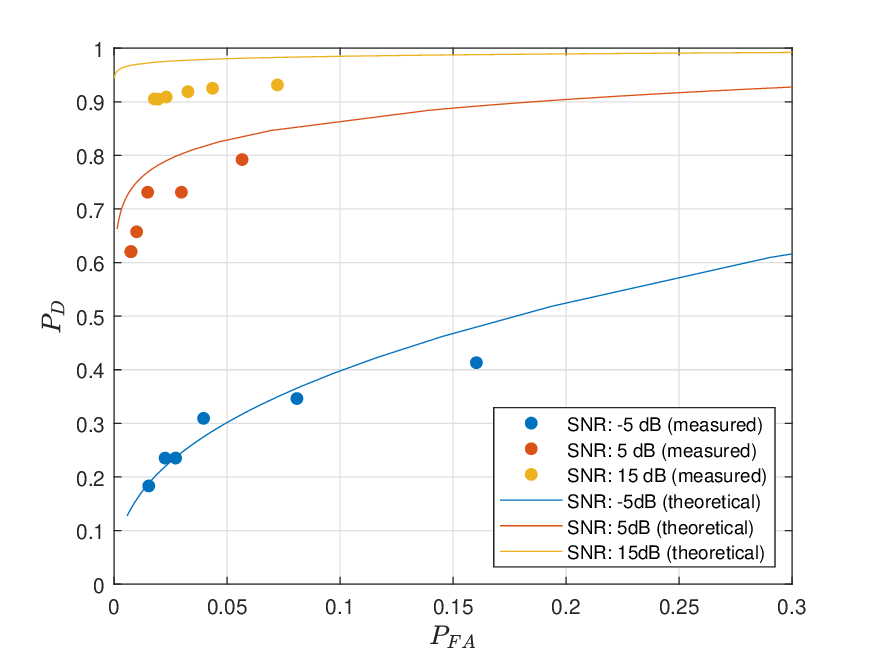}}
    \subfloat[]{\includegraphics[width=0.95\columnwidth]{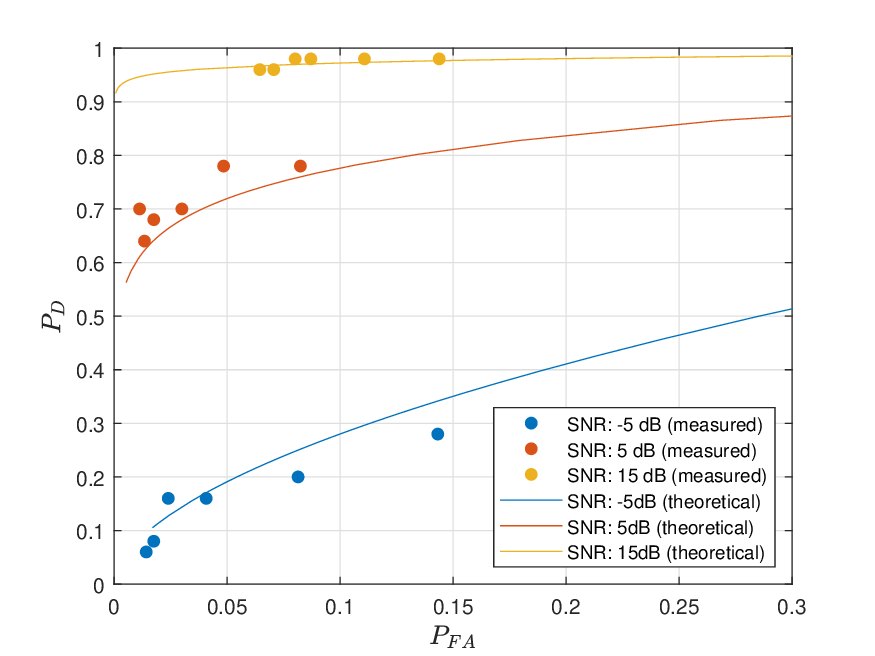}}
    \caption{ROC curves depending on SNR. (a) LoS (b) NLoS.}
    \label{fig:ROC}
\end{figure*}

In this paper, we proposed a beamspace channel denoising algorithm for mmWave massive MIMO systems, based on Bayesian binary hypothesis testing. 
The algorithm detects signal-relevant elements via a likelihood ratio test under a Bernoulli-complex Gaussian distribution and applies hard-thresholding to suppress noise. 
Its element-wise structure enables linear-time complexity while maintaining high estimation accuracy. 
Simulation results under both LoS and NLoS scenarios show that the proposed method achieves performance comparable to that of existing low-complexity or learning-based baselines, while achieving significantly lower computational complexity. 
In addition, we derived analytical expressions for the detection and false alarm probabilities, which effectively explain the estimation behavior under various conditions. 
These findings demonstrate that the proposed method may be a practical and efficient solution for real-time mmWave channel estimation, particularly in mobility-driven scenarios requiring frequent channel updates.

\appendices
\section{Derivation of Hypothesis Testing Threshold}\label{app:derivation}
In this appendix, we present the full derivation of the decision threshold in \eqref{eqn:threshold_f}, which is derived by applying the Bayesian likelihood ratio test, assuming that each beamspace element follows a Bernoulli-complex Gaussian mixture distribution. 
First, the left-hand side of \eqref{eq:test} can be rewritten as

\begin{equation}
\begin{aligned}
    &\frac{
        \dfrac{1}{\pi \left( \|\bar{\mathbf{h}}\|^2 / p(\mathcal{H}_1) + E_0 \right)}
        \exp\left( -\dfrac{|\bar{h}_m^\prime|^2}{\|\bar{\mathbf{h}}\|^2 / p(\mathcal{H}_1) + E_0} \right)
    }{
        \dfrac{1}{\pi E_0}
        \exp\left( -\dfrac{|\bar{h}_m^\prime|^2}{E_0} \right)
    } \\
    &= \frac{E_0}{\|\bar{\mathbf{h}}\|^2 / p(\mathcal{H}_1) + E_0} 
    \exp\left( 
        -\frac{|\bar{h}_m^\prime|^2}{\|\bar{\mathbf{h}}\|^2 / p(\mathcal{H}_1) + E_0}
        + \frac{|\bar{h}_m^\prime|^2}{E_0}
    \right).
\end{aligned}
\end{equation}
Subsequently, the test in \eqref{eq:test} can be reformulated as
\begin{equation}
\begin{aligned}
    &\exp\left( 
        -\frac{|\bar{h}_m^\prime|^2}{\|\bar{\mathbf{h}}\|^2 / p(\mathcal{H}_1) + E_0}
        + \frac{|\bar{h}_m^\prime|^2}{E_0}
    \right) \\
    &\quad \underset{\mathcal{H}_0}{\overset{\mathcal{H}_1}{\gtrless}} 
    \left( \frac{\|\bar{\mathbf{h}}\|^2 / p(\mathcal{H}_1) + E_0}{E_0} \right)
    \left( \frac{1 - p(\mathcal{H}_1)}{p(\mathcal{H}_1)} C \right).
\end{aligned}
\end{equation}
Applying the natural logarithm to both sides of the equation, we obtain

\begin{equation}
\begin{aligned}
    &-\frac{|\bar{h}_m^\prime|^2}{\|\bar{\mathbf{h}}\|^2 / p(\mathcal{H}_1) + E_0}
    + \frac{|\bar{h}_m^\prime|^2}{E_0} \\
    &\quad = \ |\bar{h}_m^\prime|^2\left( -\frac{1}{\|\bar{\mathbf{h}}\|^2 / p(\mathcal{H}_1) + E_0} + \frac{1}{E_0} \right)\\
    &\quad = \ |\bar{h}_m^\prime|^2 \left( \frac{\|\bar{\mathbf{h}}\|^2 / p(\mathcal{H}_1)}{(\|\bar{\mathbf{h}}\|^2 / p(\mathcal{H}_1) + E_0)E_0} \right) \\
    &\quad \underset{\mathcal{H}_0}{\overset{\mathcal{H}_1}{\gtrless}} \ 
    \ln\left( \frac{\|\bar{\mathbf{h}}\|^2 / p(\mathcal{H}_1) + E_0}{E_0} \times \frac{1 - p(\mathcal{H}_1)}{p(\mathcal{H}_1)} C \right).
\end{aligned}
\end{equation}
Rearranging the expression in terms of $|\bar{h}_m^\prime|^2$, we get
\begin{equation}
    \begin{aligned}
        |\bar{h}_m^\prime|^2 
        &\underset{\mathcal{H}_0}{\overset{\mathcal{H}_1}{\gtrless}}
        \left( \frac{(\|\bar{\mathbf{h}}\|^2 / p(\mathcal{H}_1) + E_0)E_0}{\|\bar{\mathbf{h}}\|^2 / p(\mathcal{H}_1)} \right) \\ 
        &\quad \times\ \ln\left( \frac{\|\bar{\mathbf{h}}\|^2 / p(\mathcal{H}_1) + E_0}{E_0} \times \frac{1 - p(\mathcal{H}_1)}{p(\mathcal{H}_1)} C \right).
        \\
        &\ =\ E_0\left( \frac{\|\bar{\mathbf{h}}\|^2/p(\mathcal{H}_1)}{\|\bar{\mathbf{h}}\|^2/p(\mathcal{H}_1)} + \frac{E_0}{\|\bar{\mathbf{h}}\|^2/p(\mathcal{H}_1)} \right) \\
        &\quad\times\ \ln\left( \left( \frac{\|\bar{\mathbf{h}}\|^2}{E_0p(\mathcal{H}_1)} + \frac{E_0}{E_0} \right) \times \frac{1 - p(\mathcal{H}_1)}{p(\mathcal{H}_1)} C \right) \\
        &\ =\ E_0\left( 1 + \frac{E_0}{\|\bar{\mathbf{h}}\|^2}p(\mathcal{H}_1) \right) \\
        &\quad\times\ \ln\left(\left( \frac{1}{E_0} \frac{\|\bar{\mathbf{h}}\|^2}{p(\mathcal{H}_1)} +1 \right) \frac{1-p(\mathcal{H}_1)}{p(\mathcal{H}_1)}C \right)
    \end{aligned}
\end{equation}
Subsequently, by utilizing $\mathrm{SNR}=\|\bar{\mathbf{h}}\|^2/E_0$, the right-hand side can be rewritten as

\begin{equation}
    \begin{aligned}
        &E_0\left( 1 + \frac{E_0}{\|\bar{\mathbf{h}}\|^2}p(\mathcal{H}_1) \right)
        \ln\left(\left( \frac{1}{E_0} \frac{\|\bar{\mathbf{h}}\|^2}{p(\mathcal{H}_1)} +1 \right) \frac{1-p(\mathcal{H}_1)}{p(\mathcal{H}_1)}C \right) \\
        & =\ E_0\left( \frac{p(\mathcal{H}_1)}{\mathrm{SNR}} + 1 \right)
        \ln \left( \left( 1+\frac{\mathrm{SNR}}{p(\mathcal{H}_1)} \right) \frac{1-p(\mathcal{H}_1)}{p(\mathcal{H}_1)}C \right),
    \end{aligned}
\end{equation}
which concludes the derivation.


\section{Derivation of Theoretical $P_{FA}$ and $P_D$}\label{app:pfapd}
We now derive the theoretical false alarm probability $P_{FA}$ and detection probability $P_{D}$ of the element-wise binary hypothesis test of \eqref{eq:test}.
We assume that the beamspace channel component $\bar{h}^\prime_m$ follows a circularly symmetric complex Gaussian distribution under both hypotheses.
Let the test statistic be $T=|\bar{h}^\prime_m|^2$.
As is well known, $|x|^2\sim\mathrm{Exp}(1/\sigma^2)$ for $x\sim\mathcal{CN}(0,\sigma^2)$.
Then, according to each hypothesis, $T\sim \mathrm{Exp}(1/E_0)$ under $\mathcal{H}_0$, and $T\sim\mathrm{Exp}(1/(E_0+\|\mathbf{h}\|^2/q))$ under $\mathcal{H}_1$.
The false alarm probability is defined as the probability that the test statistic exceeds the threshold $\tau$ under $\mathcal{H}_0$:
\begin{equation}
    P_{FA} = p(T>\tau | \mathcal{H}_0).
\end{equation}
Since $T\sim\mathrm{Exp}(1/E_0)$, the PDF is
\begin{equation}
    f_T(t|\mathcal{H}_0) = \frac{1}{E_0}\exp\left( -\frac{t}{E_0} \right),
\end{equation}
for $t\geq0$.
Thus, the probability is
\begin{equation}
    P_{FA} =\int_{\tau}^\infty \frac{1}{E_0}\exp\left( -\frac{t}{E_0} \right) dt = \exp\left( -\frac{\tau}{E_0} \right).
\end{equation}
In addition, the detection probability is the probability that the test statistic exceeds the threshold under $\mathcal{H}_1$:
\begin{equation}
    P_{D}=p(T>\tau|\mathcal{H}_1).
\end{equation}
Here, $T\sim\mathrm{Exp}(1/(E_0+\|\mathbf{h}\|^2/q))$, so the PDF is
\begin{equation}
    f_T(t|\mathcal{H}_1) = \frac{1}{E_0+\|\mathbf{h}\|^2/q}\exp\left( -\frac{t}{E_0+\|\mathbf{h}\|^2/q} \right),
\end{equation}
for $t\geq 0$.
Then, we obtain
\begin{equation}
    \begin{aligned}
        P_D& 
        =\int_\tau^\infty \frac{1}{E_0+\|\mathbf{h}\|^2/q}\exp\left( -\frac{t}{E_0+\|\mathbf{h}\|^2/q} \right)dt \\
        &=\exp\left(-\frac{\tau}{E_0+\|\mathbf{h}\|^2/q}\right).
    \end{aligned}
\end{equation}



\section{Proofs}
\subsection{Proof of Theorem \ref{thm:bound}}\label{app:bound}
By the definition of MSE, it can be expressed as
\begin{equation}
    \mathrm{MSE}=\frac{1}{M} \sum_{m=1}^M \mathbb{E}[|\bar{h}_m^\star-\bar{h}_m|^2].
\end{equation}
This can be separated for each hypothesis by the law of total expectation as
\begin{equation}
    \mathrm{MSE}=q\mathbb{E}[|\bar{h}_m^\star-\bar{h}_m|^2|\mathcal{H}_1]+(1-q)\mathbb{E}[|\bar{h}_m^\star-\bar{h}_m|^2|\mathcal{H}_0].
\end{equation}
First, considering $\mathcal{H}_1$, the element of estimated channel vector is
\begin{equation}
    \bar{h}_m^\star = \begin{cases}
        \bar{h}_m+\bar{e}_m & \textnormal{with probability }P_D, \\
        0 & \textnormal{with probability }1-P_{D}.
    \end{cases}
\end{equation}
Thus, the error can be calculated as
\begin{equation}
\begin{aligned}
    \mathbb{E}[|\bar{h}_m^\star-\bar{h}_m|^2|\mathcal{H}_1] & = P_D\mathbb{E}[|\bar{e}_m|^2] + (1-P_D)\mathbb{E}[|\bar{h}_m|^2] \\
    & \approx P_DE_0+(1-P_D)(\|\bar{\mathbf{h}}\|^2).
\end{aligned}
\end{equation}
In the case of $\mathcal{H}_0$, the element of estimated channel vector is
\begin{equation}
    \bar{h}_m^\star = \begin{cases}
        0 & \textnormal{with probability }1-P_{FA}, \\
        \bar{e}_m & \textnormal{with probability }P_{FA}. \\
    \end{cases}
\end{equation}
Then, the error is expressed as
\begin{equation}
    \begin{aligned}
        \mathbb{E}[|\bar{h}_m^\star-\bar{h}_m|^2|\mathcal{H}_0] & = (1-P_{FA})\cdot 0 + P_{FA} \mathbb{E}[|\bar{e}_m|^2]\\ 
        & \approx P_{FA}E_0.
    \end{aligned}
\end{equation}
Thus, the expected MSE can be approximated as
\begin{equation}
    \mathrm{MSE}\approx
    q[P_DE_0+(1-P_D)(\|\bar{\mathbf{h}}\|^2)]
    +(1-q)P_{FA}E_0.
\end{equation}

\subsection{Proof of Corollary \ref{cor:inf}}\label{app:proof_cor}
From Theorem~\ref{thm:bound}, the approximation of the expected MSE is given by:
\begin{equation}
    \mathrm{MSE} \approx q[P_DE_0+(1-P_D)\|\bar{\mathbf{h}}\|^2]+(1-q)P_{FA}E_0.
\end{equation}
Consider the asymptotic case where the detection probability $P_D$ approaches 1 and the false alarm probability $P_{FA}$ approaches 0.
In this limit, the term $(1-P_D)\|\bar{\mathbf{h}}\|^2$ vanishes, and the false alarm contribution $(1-q)P_{FA}E_0$ also converges to zero.
Therefore, the bound simplifies to 
\begin{equation}
    \mathrm{MSE}\to qE_0.
\end{equation}

\subsection{Proof of Corollary~\ref{cor:prop}\label{app:proof_cor2}}
From Theorem~\ref{thm:bound}, the approximation of the expected MSE is given by:
\begin{equation}
    \mathrm{MSE} \approx q[P_DE_0+(1-P_D)\|\bar{\mathbf{h}}\|^2]+(1-q)P_{FA}E_0.
\end{equation}
Since we consider $P_D \gg P_{FA}$, the bound can be approximated as
\begin{equation}
\begin{aligned}
     & q[P_DE_0+(1-P_D)\|\bar{\mathbf{h}}\|^2]+(1-q)P_{FA}E_0 \\
    & \approx q[P_DE_0+(1-P_D)\|\bar{\mathbf{h}}\|^2].
\end{aligned}
\end{equation}
Furthermore, we assumed $P_D \gg 1-P_D$, hence it can also be approximated as
\begin{equation}
    \begin{aligned}
        q[P_DE_0+(1-P_D)\|\bar{\mathbf{h}}\|^2]
        & = qP_DE_0+q(1-P_D)\|\bar{\mathbf{h}}\|^2 \\
        & \approx qP_DE_0 \\
        & = \beta q.
    \end{aligned}
\end{equation}

\subsection{Proof of Theorem~\ref{thm:bound_act}\label{app:proof_thm2}}
Since the estimated activity rate $\hat{q}$ is quantized into discrete set $\left\{ \frac{m}{M} | m=1,...,M \right\}$, it can be written as
\begin{equation}
    \hat{q}= \hat{q}_u + \tilde{q},
\end{equation}
with $|\tilde{q}|\leq\frac{1}{2M}$,
where $\hat{q}_u$ is the unquantized estimate given by
\begin{equation}
    \hat{q}_u = \frac{2(\widehat{\mathrm{SNR}})^2}{\dfrac{\hat{\mu}_4}{\widehat{E}_0^2}-2-4\widehat{\mathrm{SNR}}}.
\end{equation}
By decomposition, we get
\begin{equation}
    \textnormal{Var}(\hat{q}) =\textnormal{Var}\left(\hat{q}_u + \tilde{q} \right) = \textnormal{Var}(\hat{q}_u ) + \textnormal{Var}(\tilde{q}) + 2\mathrm{Cov}(\hat{q}_u,\tilde{q}).
\end{equation}
First, we bound the variance of $\hat{q}_u$.
Using delta methods, we approximate $\hat{q}_u=f(\hat{\mu}_4)$, where
\begin{equation}
    f(\mu)=\frac{2\|\bar{\mathbf{h}}\|^4}{\mu - 2E_0^2 - 4E_0\|\bar{\mathbf{h}}\|^2}.
\end{equation}
To check the availability of delta method, we first check the consistency.
By the strong law of large numbers,
\begin{equation}
    \hat{\mu}_4=\frac{1}{M}\sum_{m=1}^M |\bar{h}_m^\prime|^4 \overset{\textnormal{a.s.}}{\longrightarrow} \mu_4.
\end{equation}
Accordingly, it is the consistent estimator.
Also, since $\{|\bar{h}_m^\prime|^4\}$ is i.i.d. and has finite variance because it is mixture of Gaussians, it satisfies
\begin{equation}
    \sqrt{M}(\hat{\mu}_4-\mu_4) \overset{D}{\longrightarrow} \mathcal{N}(0,\sigma^2),
\end{equation}
where $\sigma^2>0$ is an arbitrary constant. Thus, $\textnormal{Var}(\hat{\mu}_4)=\sigma^2/M$.
Also, $f(\mu)$ is differentiable in a neighborhood of $\mu_4$.
Hence, delta method is applicable, and provides a valid first-order approximation of the variance of $\hat{q}_u$ as follows:
\begin{equation}
    \begin{aligned}
    f^\prime(\mu_4) & = -\frac{2\|\bar{\mathbf{h}}\|^4}{(\mu_4 - 2E_0^2 - 4E_0\|\bar{\mathbf{h}}\|^2)^2} \\
    & = -\frac{2\|\bar{\mathbf{h}}\|^4}{(2\|\bar{\mathbf{h}}\|^4/q)^2} \\
    & = -\frac{q^2}{2\|\bar{\mathbf{h}}\|^4},
    \end{aligned}
\end{equation}
because $\mu_4 = 2E_0^2 + 4E_0\|\bar{\mathbf{h}}\|^2 + 2\|\bar{\mathbf{h}}\|^4/q$.
Therefore, the variance of $\hat{q}_u$ can be simplified as
\begin{equation}
    \textnormal{Var}(\hat{q}_u) = (f^\prime(\mu_4))^2\cdot \textnormal{Var}(\hat{\mu}_4) = \frac{cq^4}{M},
\end{equation}
where $c=\sigma^2/4\|\bar{\mathbf{h}}\|^8$.
Since $q\in\{\frac{m}{M}|m=1,...,M\}$ by definition, it follows that $q^4\leq q^2$.
Also, since $q\geq\frac{1}{M}$,
\begin{equation}
    q^4 = q^2\cdot q^2 \geq q^2 \cdot \left(\frac{1}{M}\right)^2 = \frac{q^2}{M^2}.
\end{equation}
Hence, there exist constants $c_1^\prime, c_2^\prime>0$ such that:
\begin{equation}
    \frac{c_1^\prime q^2}{M} \leq \textnormal{Var}(\hat{q}_u) \leq \frac{c_2^\prime q^2}{M}.
\end{equation}
Meanwhile, utilizing the inequality, the variance of $\tilde{q}$ is bounded as
\begin{equation}
    \textnormal{Var}(\tilde{q})\leq \left( \frac{1}{2M} \right)^2 = \frac{1}{4M^2}.
\end{equation}
Also, using the Cauchy-Schwarz inequality, the covariance between $\hat{q}$ and $\tilde{q}$ is bounded by
\begin{equation}
    |\mathrm{Cov}(\hat{q},\tilde{q})| \leq \sqrt{\textnormal{Var}(\hat{q})\cdot\textnormal{Var}(\tilde{q})}
    \leq \sqrt{\frac{c_2^\prime q^2}{M}\cdot \frac{1}{4M^2}}=\frac{c_3 q}{M^{3/2}}.
\end{equation}
Consequently, we now combine the terms for upper bound as
\begin{equation}
\begin{aligned}
    \textnormal{Var}(\hat{q}) & = \textnormal{Var}(\hat{q}_u ) + \textnormal{Var}(\tilde{q}) + 2\mathrm{Cov}(\hat{q}_u,\tilde{q}) \\
    & \leq \frac{c_2^\prime q^2}{M} + \frac{1}{4M^2} + \frac{2c_3 q}{M^{3/2}} 
    \\ & 
    \leq \frac{c_2 q^2}{M},
\end{aligned}
\end{equation}
since we assumed that $M$ is sufficiently large.
Likewise, the lower bound can be derived as
\begin{equation}
    \textnormal{Var}(\hat{q}) \geq \frac{c_1 q^2}{M} - \frac{1}{4M^2} - \frac{2c_3 q}{M^{3/2}}
    \geq \frac{c_1 q^2}{M}.
\end{equation}
Also, this lower/upper bounds imply that 
\begin{equation}
    \textnormal{Var}(\hat{q})\propto q^2.
\end{equation}

\end{document}